\documentclass[conference]{IEEEtran}
\usepackage[utf8]{inputenc}
\usepackage{amsmath}
\usepackage{amsthm}
\usepackage{amssymb}
\usepackage{url}
\usepackage{fourier}

\newtheorem{lemma}{Lemma}

\title{A Packing Lemma for Polar Codes}
\author{\IEEEauthorblockN{Erdal Ar{\i}kan}
\IEEEauthorblockA{
	Bilkent University\\ Ankara, Turkey\\
Email: arikan@ee.bilkent.edu.tr}}

\def\cK{\mathcal{K}}

\def\cM{\mathcal{M}}

\def\cA{\mathcal{A}}

\def\cT{\mathcal{T}}
\def\cS{\mathcal{S}}
\def\cX{\mathcal{X}}
\def\cH{\mathcal{H}}
\def\cF{\mathcal{F}}
\def\cX{\mathcal{X}}
\def\cY{\mathcal{Y}}
\def\cC{\mathcal{C}}

\def\cE{\mathcal{E}}
\def\cD{\mathcal{D}}
\def\cB{\mathcal{B}}

\newcommand{\defn}{\mathrel{\stackrel{\Delta}{=}}}

\def\1{\mathbb{1}}
\def\F{\mathbb{F}}

\DeclareMathOperator{\supp}{supp}

\def\twonr{2^{\lceil nR\rceil}}
\def\tset{\cT_\epsilon^{(n)}}
\def\eqdot{\dot{=}}
\def\tx{\tilde{x}}
\def\ty{\tilde{y}}
\def\ts{\tilde{s}}
\def\ox{\overline{x}}

\begin{document}
\maketitle

\begin{abstract}
A packing lemma is proved using a setting where the channel is a binary-input discrete memoryless channel $(\mathcal{X},w(y|x),\mathcal{Y})$, the code is selected at random subject to parity-check constraints, and the decoder is a joint typicality decoder. The ensemble is characterized by (i) a pair of fixed parameters $(H,q)$ where $H$ is a parity-check matrix and $q$ is a channel input distribution and (ii) a random parameter $S$ representing the desired parity values. For a code of length $n$, the constraint is sampled from $p_S(s) = \sum_{x^n\in \cX^n} \phi(s,x^n)q^n(x^n)$ where $\phi(s,x^n)$ is the indicator function of event $\{s = x^n H^T\}$ and $q^n(x^n) = \prod_{i=1}^nq(x_i)$. Given $S=s$, the codewords are chosen conditionally independently from $p_{X^n|S}(x^n|s) \propto \phi(s,x^n) q^n(x^n)$. It is shown that the probability of error for this ensemble decreases exponentially in $n$ provided the rate $R$ is kept bounded away from $I(X;Y)-\frac{1}{n}I(S;Y^n)$ with $(X,Y)\sim q(x)w(y|x)$ and $(S,Y^n)\sim p_S(s)\sum_{x^n} p_{X^n|S}(x^n|s) \prod_{i=1}^{n} w(y_i|x_i)$. In the special case where $H$ is the parity-check matrix of a standard polar code, it is shown that the rate penalty $\frac1nI(S;Y^n)$ vanishes as $n$ increases. The paper also discusses the relation between ordinary polar codes and random codes based on polar parity-check matrices.
\end{abstract}
\IEEEpeerreviewmaketitle

\section{Introduction}

Packing and covering lemmas are basic building blocks of coding theorems in information theory.
The book by El Gamal and Kim \cite{EGK} exemplifies this; it relies on a small number of packing and covering lemmas (such as Lemma 3.1 \cite[p.~46]{EGK} and Lemma 3.3 \cite[p.~64]{EGK}) to prove a vast number of coding theorems for multi-terminal source and channel coding problems.
Unfortunately, the packing and covering lemmas used for proving theorems in a clean way rely on joint, or at least pairwise, independence among the codewords.
Joint or pairwise independence are too strong assumptions for various practical code ensembles, including those for polar codes.
The goal of this paper is to prove a packing lemma under less stringent conditions on the code ensemble.
The motivation behind this work is to develop packing and covering lemmas that are applicable to polar codes so that existing proofs based on standard code ensembles can be translated readily to similar proofs for polar codes. In this paper, we address only the packing problem. The results are preliminary. More work is needed to establish the desired links between random-coding methods and explicit polar code constructions.

In Sect.~\ref{Sect:Standard}, we review the random-coding method in the absence of any constraints. In Sect.~\ref{Sect:Constraints}, we extend the method of Sect.~\ref{Sect:Standard} to the case of random-coding subject to parity-check constraints.
In Sect.~\ref{Sect:Polar}, we further specialize the results to the case of parity-check matrices obtained from polar coding. The paper concludes in Sect.~\ref{Sect:Remarks} with a summary and remarks.

\section{Standard random-coding method}\label{Sect:Standard}
This section reviews the standard random-coding method.
We follow the presentation given in \cite[Sect.~3.1.2]{EGK} and, for the most part, adopt the notation and conventions there. 

Consider a communication system employing block coding over a discrete memoryless channel (DMC) $(\cX,w(y|x),\cY)$ with input alphabet $\cX$, output alphabet $\cY$, and transition probabilities $w(y|x)$, $x\in \cX$, $y\in \cY$.
Let $R$ denote the code rate, $n$ the length of the codewords, and $c=\{x^n(1),\ldots,x^n(\twonr)\}$
the code itself.
To send message $m$, one transmits the codeword $x^n(m)$ into the channel; in response, the channel
outputs a word $y^n$ with probability 
\begin{equation}\label{Eq:Channel}
w^n(y^n|x^n(m)) \defn \prod_{i=1}^n w(y_i|x_i(m));
\end{equation}
and, the decoder in the system maps $y^n$ to a decision $\hat{m}\in [1:\twonr] \cup \{e\}$ where $e$ is a special symbol indicating decoder failure. Here, the decoder is assumed to be a joint typicality decoder designed for a channel input-output ensemble $(X,Y)\sim q(x) w(y|x)$ where $q(x)$ is a given probability distribution on $\cX$. 
Given $y^n$, the joint typicality decoder outputs $\hat{m}(y^n)= j$ if $j$ is the unique message index in $[1:\twonr]$ such that 
$(x^n(j),y^n)\in \tset(X,Y)$; otherwise, the output is $\hat{m}=e$. 
Here, $\tset$ is defined as in \cite[p.~27]{EGK}, namely, as the set of all $(x^n, y^n)\in \cX^n\times \cY^n$ such that the inequalities 
$$
|\pi(x, y|x^n, y^n) - q(x)w(y|x)| \le \epsilon\, q(x)w(y|x)
$$
hold for each $(x, y) \in  \cX \times \cY$, where $\pi(x,y|x^n,y^n)$ is the fraction of times $(x,y)$ appears as a coordinate of $(x^n,y^n)$. 

In random-coding analysis of such a system, one regards the code $c$ as a sample of a random code $\cC$, drawn with probability
\begin{equation}\label{Eq:RC}
p_{\cC}(c) = \prod_{j=1}^{\twonr} q^n(x^n(j)),
\end{equation}
where $x^n(j)$ denotes the $j$th codeword in $c$ and $q^n(x^n)\defn \prod_{i=1}^nq(x_i)$.
The entire system is represented by an ensemble $(M,\cC,Y^n,\hat{M})$ with a probability assignment $p_{M,\cC,Y^n,\hat{M}}(m,c,y^n,\hat{m})$ of the form
\begin{equation}\label{Eq:Ensemble}
p_M(m)\, p_{\cC}(c)\,p_{Y^n|M,\cC}(y^n|m,c)\, p_{\hat{M}|\cC,Y^n}(\hat{m}|c,y^n),
\end{equation}
where $p_M(m)$ is uniform on $[1:\twonr]$, $p_{Y^n|M,\cC}(y |m,c)$ is given by \eqref{Eq:Channel} with $x^n(m)$ as the $m$th codeword of $c$, and $\hat{M}$ is a function of $(\cC,Y^n)$ as determined by the operation of the joint typicality decoder. 

Let $\cE = \{\hat{M}\neq M\}$ denote the error event and $P(\cE)$ the probability of error w.r.t. the above ensemble.
The goal of the random coding analysis is to show that, for any fixed $R < I(X;Y)$ with $(X,Y)\sim q(x)w(y|x)$, the probability of error $P(\cE)$ goes to zero as the block-length $n$ increases.
The analysis begins by observing that, due to symmetry,
$P(\cE) = P(\cE|M=1)$.
Then, one defines
$\cE_1 = \{(X^n(1),Y^n)\notin \tset\}$ and
$\cE_2 = \{(X^n(j),Y^n)\in \tset \ \text{for some $j\neq 1$}\}$, so that one can write
$P(\cE|M=1) = P(\cE_1 \cup \cE_2|M=1)\le P(\cE_1|M=1) + P(\cE_2|M=1)$.
By standard results in large-deviation analysis, it is observed that $P(\cE_1|M=1)$ goes to 0 (exponentially) in $n$.
For the second term, the union bound is used to write
\begin{align}
P(\cE_2|M=1) \le \sum_{j=2}^{\twonr} P(\cD_j|M=1)
\end{align}
where $\cD_j \defn \{(X^n(j),Y^n)\in \tset\}$;
then, a joint typicality lemma is invoked to bound each term in the union bound as 
\begin{equation}
P(\cD_j|M=1) \,\eqdot \,2^{-nI(X;Y)}, \quad j\neq 1,
\end{equation}
which establishes that $P(\cE_2|M=1) \eqdot 2^{n(R-I(X;Y))}$.
This completes the proof that $P(\cE)$ goes to zero (exponentially) in $n$ provided $R <I(X;Y)$.
If one chooses $q(x)$ as a distribution that maximizes $I(X;Y)$, one obtains a proof of achievability of the channel capacity $C \defn \max_{q(x)} I(X;Y)$.

\section{Random coding under constraints}\label{Sect:Constraints}
In this section, we consider the same channel coding problem as in Sect.~\ref{Sect:Standard} with the difference that here the code ensemble $\cC$  is subject to certain constraints. The target application of the method developed in this section is polar coding; however, for broader applicability and a wider perspective, initial formulation is given in a fairly general manner.

\subsection{Code generation under constraints}
The constraints on code generation will be represented by a parameter $s$ taking values over a space $\cS$.
We will consider codes of length $n$ and let $x^n\in \cX^n$ denote a generic channel input word of length $n$.
We will model the constraints by a function $\phi:\cS\times \cX^n \to \{0,1\}$
such that $\phi(s,x^n)=1$ iff $x^n$ satisfies the constraint $s$.
As a simple example, let $\cS =\{o,e\}$ and let $\phi(e,x^n) = 1$ iff the parity of $x^n$ is even and $\phi(o,x^n)=1$ iff the parity of $x^n$ is odd. A more general parity-check constraint will be treated in the next section.

We will say that a constraint functions $\phi$ is symmetric if there exists non-zero reals $(\alpha_s: s\in \cS)$ such that 
\begin{equation}\label{Eq:ConstraintSymmetry}
\sum_{s\in \cS} \alpha_s \phi(s,x^n) = 1, \quad \text{for all $x^n\in \cX^n$}.
\end{equation}
For example, the odd-even parity constraint is symmetric with $\alpha_s = 1$.
We will restrict attention to symmetric constraint functions.

The random code ensembles that we will consider will be denoted as $(S,\cC)$ with $S$ denoting a random constraint variable that takes values in $\cS$ and $\cC=\{X^n(1),\ldots,X^n(\twonr)\}$ denoting a code chosen at random subject to the constraint $S$. 
We take $q(x)$, the target channel input distribution, as given.
For any particular constraint $s\in \cS$ and code $c = \{x^n(1),\ldots,x^n(\twonr)\}$, we specify the probability assignment on $(S,\cC)$ as  
\begin{align}\label{Eq:RC2}
p_{S,\cC}(s,c) & = p_S(s) \prod_{m=1}^{\twonr} q_s(x^n(m))
\end{align}
where
\begin{equation}\label{Eq:rs}
p_S(s) \defn \alpha_s \sum_{x^n} \phi(s,x^n) q^n(x^n), \quad s\in \cS,
\end{equation}
and 
\begin{equation}\label{Eq:qs}
q_s(x^n) \defn \frac{\phi(s,x^n) q^n(x^n)}{\sum_{\tx^n} \phi(s,\tx^n)q^n(\tx^n)} ,\quad x^n\in \cX^n.
\end{equation}
Thus, the codewords $\{X^n(m)\}$ are selected in a conditionally i.i.d. manner from $q_s$, given the constraint $S=s$.
Note that the marginal distribution of individual codewords is given by 
\begin{equation}\label{Eq:rsqs}
p_{X^n(m)}(x^n) = \sum_s p_S(s) q_s(x^n) = q^n(x^n), \quad x^n\in \cX^n,
\end{equation}
which is in agreement with the target channel-input distribution.
Also note that the channel output follows a product-form distribution
\begin{equation}
p_{Y^n}(y^n) = t^n(y^n) \defn \prod_{i=1}^n t(y_i)
\end{equation}
with $t(y) \defn \sum_x q(x)w(y|x)$.

\subsection{Analysis of probability of error}
We now analyze the average performance of the constrained code ensemble defined by \eqref{Eq:RC2}.
As in Sect.~\ref{Sect:Standard}, we assume that the message random variable $M$ is uniformly distributed over $[1:\twonr]$ and that a joint typicality decoder is being used.
The joint ensemble for the system will be $(M,S,\cC,Y^n,\hat{M})$ with a probability assignment 
\begin{equation}\label{Eq:Ensemble2}
p_M(m)\, p_{S,\cC}(s,c)\,p_{Y^n|M,\cC}(y^n|m,c)\, p_{\hat{M}|\cC,Y^n}(\hat{m}|c,y^n),
\end{equation}
which is the same as \eqref{Eq:Ensemble}, except here the code ensemble is defined by \eqref{Eq:RC2}.
A property of this ensemble, which will be important in the sequel, is the independence of $(S,Y^n)$ and $M$.
This can be verified by writing
\begin{align*}
p_{S,Y^n|M}(s,y^n|m) & = \sum_{x^n} p_{S,X^n(m),Y^n|M}(s,x^n,y^n|m)\\
& = \sum_{x^n} p_S(s) q_s(x^n)w^n(y^n|x^n),
\end{align*}
and observing that the final sum is independent of $m$.

We now turn to the error analysis and define the error events $\cE$, $\cE_1$, $\cE_2$ as in Sect.~\ref{Sect:Standard}.
As before, by symmetry, we have $P(\cE) \le P(\cE_1|M=1)+P(\cE_2|M=1)$.
As in Sect.~\ref{Sect:Standard}, the first term $P(\cE_1|M=1)$ goes to zero exponentially in $n$. 
To bound the second term $P(\cE_2|M=1)$, we will use an argument involving the sets $\cD_j$ as defined in Sect.~\ref{Sect:Standard},
as well as the mutual information random variable
\begin{equation}\label{def:mutinf}
i(s;y^n)=\log\frac{p_{S,Y^n}(s,y^n)}{p_{S}(s)p_{Y^n}(y^n)} =\log\frac{p_{S,Y^n}(s,y^n)}{p_S(s)t^n(y^n)},
\end{equation}
and the event 
\begin{equation}\label{def:A}
\cA =\{i(S;Y^n) > n\gamma\}.
\end{equation}
The $\gamma$ in the definition of $\cA$ is a real number that will be specified later.
In terms of these, we have the following bound.
\begin{align*}
P(\cE_2|M=1) & = P(\cE_2\cap \cA|M=1) + P(\cE_2\cap \cA^c|M=1)\\
& \le P(\cA|M=1) + \sum_{j=2}^{\twonr} P(\cD_j\cap \cA^c|M=1)\\
& = P(\cA) + (\twonr-1) P(\cD_2\cap \cA^c|M=1),
\end{align*}
where in the last line we replaced $P(\cA|M=1)$ with $P(\cA)$ by noting that $\cA$, being an event defined in terms of $(S,Y^n)$, is independent of $\{M=1\}$.
We define $\cB$ as the set of all $(s,x^n,y^n)\in \cS\times \cX^n\times \cY^n$ such that $(x^n,y^n)\in \tset$ and $i(s;y^n) \le n\gamma$, and continue as follows.
\begin{align*}
P(\cD_2\cap \cA^c|M=1)
& = \sum_{(s,x^n,y^n)\in\cB} p_{S,Y^n}(s,y^n)q_s(x^n)\\
& \stackrel{(a)}{\le}  \sum_{(s,x^n,y^n)\in \cB }2^{n\gamma} p_S(s)t^n(y^n) q_s(x^n)\\
& \stackrel{(b)}{\le} \sum_{(s,x^n,y^n)\in \cS\times\tset} 2^{n\gamma}p_S(s)t^n(y^n) q_s(x^n)\\
& \stackrel{(c)}{=} \sum_{(x^n,y^n)\in \tset} 2^{n\gamma} t^n(y^n) q^n(x^n)\\
& \stackrel{(d)}{\eqdot}  2^{-n(I(X;Y)- \gamma)}
\end{align*}
where (a) follows by the fact that, for any $(s,x^n,y^n)\in \cB$, $p_{S,Y^n}(s,y^n) \le 2^{n\gamma} \,p_S(s)t^n(y^n)$, 
(b) by extending the range of the sum from $\cB$ to the larger set $\cS\times \tset$,
(c) by carrying out the sum over $s\in \cS$,
and (d) by the joint typicality lemma \cite[p.~43]{EGK}. 
Collecting the results, we have the bound
$$
P(\cE_2|M=1) \le P(\cA) + 2^{n (R-I(X;Y)+\gamma)}. 
$$
To keep the upperbound on $P(\cE_2|M=1)$ under control, we need a large enough $\gamma$ so that $P(\cA)$ is small, but also a rate $R$ smaller than $I(X;Y)-\gamma$. 
These two conflicting objectives put into evidence that there is a trade-off between performance and structure.
For a more quantitative asymptotic statement, consider a sequence of ensembles $\{(S_n,\cC_n)\}$ with each ensemble in the sequence having the same code rate $R$. Let $P_{e,n}$ denote
the probability of error for the $n$th ensemble. Let  
\begin{equation}\label{Eq:gamma}
\gamma^* =\inf\left\{ \gamma: \limsup_{n\to \infty} P\left(i(S_n;Y^n) > n\gamma\right) = 0\right\}.
\end{equation}
Then, $P_{e,n}$ goes to zero if $R < I(X;Y) -\gamma^*$.
If the sequence $\{(S_n,\cC_n)\}$ has a convergence property such as 
$$
\limsup_{n\to \infty} \left\{P\left(|i(S_n;Y^n) - I(S_n;Y^n)|\ge n\epsilon\right) \right\}= 0,
$$
for any fixed $\epsilon >0$, then we may take 
\begin{equation}\label{Eq:gamma2}
\gamma^*= \limsup_{n\to\infty} \left\{ \frac{1}{n}I(S_n;Y^n)  \right\}.
\end{equation}
In any case, it is apparent that the cost of placing constraints on the code is a rate penalty given by $\gamma^*$.
We summarize the above discussion as follows.
\begin{lemma}\label{Lemma:Constraints}
Let $\{(S_n,\cC_n)\}$ be a sequence of constrained code ensembles indexed by code length $n$, with each ensemble in the sequence defined by \eqref{Eq:RC2} and having a common rate $R$.
Let $P_{e,n}$ denote the probability of error for the $n$th ensemble, under joint typicality decoding.
Then, $P_{e,n}$ goes to zero as $n$ increases provided $R <I(X;Y)-\gamma^*$ where $\gamma^*$ is defined by \eqref{Eq:gamma}.
\end{lemma}

\subsection{Parity-check constraints}\label{Sect:ParityCheck}
In this part, we continue the above discussion for the important special case of parity-check constraints.
For simplicity, we restrict the discussion to channels with binary input alphabets, $\cX =\{0,1\}$. 
We will identify $\cX$ with the binary field $\F_2$ and use vector space operations over $\F_2$ to define the code constraints.
The joint ensemble for the system will still be $(M,S,\cC,Y^n,\hat{M})$ with a probability assignment 
\eqref{Eq:Ensemble2}, except here we will consider a constraint function $\phi$ defined in terms of a parity-check matrix $H\in \F_2^{r\times n}$ with 
$r$ rows and $n$ columns. We leave $r$ as an arbitrary parameter, $0\le r\le n$, through the following analysis and discuss its effect on the results following the analysis.
We take the constraint set as $\cS = \F_2^{r}$ and for any $(s,x^n)\in \cS\times \cX^n$ define the constraint function as
\begin{equation}\label{Eq:Parity}
\phi(s,x^n) =\begin{cases} 1, & \text{if $s = x^nH^T$},\\
0, & \text{otherwise}.
\end{cases}
\end{equation}
Note that $\phi$ is symmetric with $\alpha_s = 1$ for every $s\in \cS$.
Also note that $\phi$ splits the set $\cX^n$ into cosets $\cK_s\defn\{x^n\in\cX^n: x^n H^T = s\}$ indexed by $s\in \cS$.
Each coset has $|\cK_s|= 2^{n-r}$ elements and $\cK_s = x_s^n + \cK_0$ where $x_s^n\in \cK_s$ is a coset representative for $\cK_s$ and $\cK_0$ denotes the coset for $s=0^r$.  

\begin{lemma}\label{Lemma:CM}
Let $\cA$ be as in \eqref{def:A} with $\gamma = \frac{1}{n}I(S;Y^n) + \epsilon$ for some $\epsilon >0$.
Then, for the parity-check code ensemble,
\begin{equation}\label{Eq:CM}
P(\cA) \le \exp\left(-n\frac{2\epsilon^2}{d}\right),
\end{equation}
where $d$ is a constant determined by $q(x)$ and $w(y|x)$.
\end{lemma}

\begin{proof}
Note that $i(S;Y^n)=f(X^n,Y^n)$ where $f(x^n,y^n) \defn i(x^nH^T;y^n)$. Writing $i(S;Y^n)$ in this way as a function of $(X^n,Y^n)$ is useful because the function $f$ is Lipschitz:
Let $(x^n,y^n)\in \cX^n\times \cY^n$ and $(\tilde{x}^n,\tilde{y}^n)\in \cX^n\times \cY^n$ be any two points such that
(a) $(x_i,y_i) \neq (\tx_i,\ty_i)$ for some $i\in [1:n]$ but $(x_j,y_j) = (\tx_j,\ty_j)$ for all $j\neq i$, $1\le j\le n$,
and (b) $q^n(x^n)w^n(y^n|x^n) >0$ and $q(\tx^n)w^n(\ty^n|\tx^n) >0$. 
We claim that 
\begin{equation}\label{Eq:Lipschitz}
\left| f(x^n,y^n)-f(\tx^n,\tx^n)\right| \le d_i,
\end{equation}
for some constant $d_i$ that depends only on the distributions $q(x)$ and $w(y|x)$. 

Assuming for a moment that the claim \eqref{Eq:Lipschitz} is true, the lemma follows from Azuma-Hoeffding inequality, specifically, from the form of this inequality as given in \cite[Corol.~5.2]{CM},
with $d=\frac{1}{n}\sum_{i=1}^n d_i^2$.
Therefore, it suffices to prove only \eqref{Eq:Lipschitz}, or equivalently,
$$
2^{-d_i} \le 2^{f(x^n,y^n)-f(\tx^n,\tx^n)}\le 2^{d_i}.
$$
To that end, we write
$$
2^{f(x^n,y^n)-f(\tx^n,\tx^n)} = \left(\frac{p_{S,Y^n}(s,y^n)}{p_{S,Y^n}(\ts,\ty^n)}\right)\left(\frac{p_S(\ts)}{p_S(s)}\right)\left(\frac{p_{Y^n}(\ty^n)}{p_{Y^n}(y^n)}\right),
$$
where we put for shorthand $s \defn x^nH^T$, $\ts \defn \tx^nH^T$.
Using the coset structure of the constraints, we have
\begin{align*}
p_{S,Y^n}(s,y^n) &= \sum_{\ox^n\in \cX^n} p_S(s) q_s(\,\ox^n) w^n(y^n|\ox^n) \\
& = \sum_{\ox^n\in \cX^n} \phi(s,\ox^n)q^n(\,\ox^n)w^n(y^n|\ox^n) \\
& =\sum_{\ox^n\in \cK_s} q^n(\,\ox^n)w^n(y^n|\ox^n) \\
& =\sum_{\ox^n\in\cK_0} q^n(x^n+\ox^n)w^n(y^n|x^n+\ox^n).
\end{align*}
Thus, we have
\begin{align*}
\frac{p_{S,Y^n}(s,y^n)}{p_{S,Y^n}(\ts,\ty^n)} = \frac{\sum_{\ox^n\in\cK_0} q^n(x^n+\ox^n)w^n(y^n|x^n+\ox^n)}{\sum_{\ox^n\in\cK_0} q^n(\tx^n+\ox^n)w^n(\ty^n|\tx^n+\ox^n)}.
\end{align*}
Now, term by term, we have the bound
\begin{align*}
\frac{q^n(x^n+\ox^n)w^n(y^n|x^n+\ox^n)}{q^n(\tx^n+\ox^n)w^n(\ty^n|\tx^n+\ox^n)} =\frac{q(x_i+\ox_i)w(y_i|x_i+\ox_i)}{q(\tx_i+\ox_i)w(\ty_i|\tx_i+\ox_i)}\le \beta_{q,w}
\end{align*}
where 
$$
\beta_{q,w} \defn \frac{\max\left\{q(x)w(y|x): (x,y)\in \supp(q(x)w(y|x))\right\}}{\min\left\{q(x)w(y|x): (x,y)\in \supp(q(x)w(y|x))\right\}},
$$
where ``$\supp$'' denotes the support of a distribution.
So,
$$
(\beta_{q,w})^{-1} \le \frac{p_{S,Y^n}(s,y^n)}{p_{S,Y^n}(\ts,\ty^n)}\le \beta_{q,w}.
$$
Using the same type of argument, we get
$$
(\beta_q)^{-1} \le \frac{p_{S}(\ts)}{p_{S}(s)}\le \beta_q, \qquad (\beta_t)^{-1} \le \frac{p_{Y^n}(\ty^n)}{p_{Y^n}(y^n)}\le \beta_t.
$$
where $\beta_q$ is defined as the ratio of $\max\{q(x): x\in \supp(q(x))\}$ to $\min\{q(x): x\in \supp(q(x))\}$
and $\beta_t$ as the ratio of $\max\{t(y): y\in \supp(t(y))\}$ to $\min\{t(y): y\in \supp(t(y))\}$.
Combining these, we obtain the proof of \eqref{Eq:Lipschitz} with 
$
d_i = \log_2 \left(\beta_{q, w}\beta_q \beta_{t}\right).
$
The lemma follows, with $d= \left(\log_2 \left(\beta_{q, w}\beta_q \beta_{t}\right)\right)^2$.
\end{proof}
This shows that $P(\cA)$ goes to zero exponentially in $n$ regardless of the size (number of rows $r$) and form of $H$;
it should be clear, however, that the specific form of $H$ affects the rate penalty $\frac{1}{n}I(S;Y^n)$.
To gain a more intuitive understanding of this issue, let us interpret $I(S;Y^n)$ as the average information leaked by the received word $Y^n$ about the constraint $S$ in a one shot transmission scenario where a codeword $X^n$ satisfying the constraint $\phi(S,X^n)=1$ is sent. 
From this perspective, we may expect that the larger the number of parity checks and the more sparse they are (involving fewer codeword digits), the larger will be the leakage.
As a trivial example, we have $H=I_n$ (the identity matrix) with $I(S;Y^n) = I(X^n;Y^n) = nI(X;Y)$, corresponding to maximum information leakage.
A non-trivial example in the same vein is Gallager's proof \cite[\S 3.8]{LDPC} that $I(S;Y^n)$ is bounded away from zero when $H$ is the parity-check matrix of a regular LDPC code of a given rate.
At the other extreme, we have the well-known fact that random parity-check codes achieve capacity, which {\sl a fortiori} implies that $I(S;Y^n)$ is typically $o(n)$.

\section{Polar parity-check matrices}\label{Sect:Polar}

In this part, we apply the results of Sect.~\ref{Sect:ParityCheck} to the situation where $H$ is a parity-check matrix derived from polar coding and show that there is no rate penalty in this case.
For brevity, we will refer to parity-check matrices obtained from polar coding as ``polar parity-check'' matrices.
We first give a brief description of polar codes; for details, we refer to \cite{ArikanIT2009}.
Let $F=\left[\begin{smallmatrix} 1 & 0 \\ 1 & 1 \end{smallmatrix}\right]$ and $G_\ell = F^{\otimes \ell}$ denote
the $\ell$th Kronecker power of $F$. Note that $G_\ell$ is an $n\times n$ matrix with $n=2^\ell$ and its inverse is itself, $G_\ell^{-1}= G_\ell$.
Polar codes are defined in terms of the mapping $x^n = u^n G_\ell$ where $x^n$ denotes the codeword and $u^n$ denotes the source word.
In polar coding we ``freeze'' a certain subset of coordinates of the source word $u^n$ and insert the data payload in the remaining portion of $u^n$.
To be specific, let $\cF\subset [1:n]$ denote the indices marking the frozen part of $u^n$ and let $u_{\cF}=(u_i:i\in \cF)$ denote the frozen part.
By convention, we set $u_{\cF} = s$ for some fixed pattern $s\in \cX^{|\cF|}$ and keep this part unchanged from one transmission to next, while we leave the other part $u_{\cF^c}$ free.
The parity-check matrix for polar codes can be derived as follows.
We begin with the definition that a word $x^n$ is a polar codeword iff $x^n = u^nG_\ell$ for some $u^n$ with $u_{\cF}=s$.
Using the inverse relation $u^n = x^nG_\ell^{-1}$, we obtain that $x^n$ is a codeword iff $s = \left(x^n G_\ell^{-1}\right)_{\cF}$.
Next, we observe that
$$
\left(x^n G_\ell^{-1}\right)_{\cF} = x^n \left(G_\ell^{-1}\right)_{\cF}
$$
where $\left(G_\ell^{-1}\right)_{\cF}$ denotes the submatrix of $G^{-1}_\ell$ obtained by taking the columns with indices in $\cF$.
Thus, we obtain a parity-check matrix for polar codes, namely,
\begin{equation}\label{Eq:Association}
H =  \left(\left(G_\ell^{-1}\right)_{\cF}\right)^T.
\end{equation}
Now, we consider Lemma~\ref{Lemma:CM} in connection with an ensemble $(S,X^n,Y^n)$ based on a polar parity-check matrix.
We annex to this ensemble the random vector $U^n\defn X^nG^{-1}_\ell$ that corresponds to the source word in polar coding
so that we have the relation
$$
S = \left(X^nG_\ell^{-1}\right)_{\cF} = U_{\cF}.
$$
We wish to show that if $\cF$ is chosen using the usual polar code design rules, then the rate penalty $I(S;Y^n)$ will be negligible.
The specific design rule we use here fixes a $\beta <1/2$ and selects the frozen set as
\begin{equation}\label{Eq:FrozenSet}
\cF = \left\{i\in [1:n]: H(U_i|Y^n,U^{i-1}) > 2^{-n^\beta}\right\}.
\end{equation}
Now, by standard facts about the entropy function, we have
\begin{align*}
I(U_\cF;Y^n) & \stackrel{(a)}{=} \sum_{i\in \cF} I(U_i;Y^n|U_{\cF_{i-1}})\\
 & = \sum_{i\in \cF} [H(U_i|U_{\cF_{i-1}}) - H(U_i|Y^n,U_{\cF_{i-1}})]\\
 & \le \sum_{i\in \cF} [1 - H(U_i|Y^n,U^{i-1})]\\
& \stackrel{(b)}{\le} |\cM| + \sum_{i\in \cH} 2^{-n^\beta}\\
& \stackrel{(c)}{\le} o(n) + n 2^{-n^\beta} = o(n)
\end{align*}
where in (a) we defined $\cF_{i-1}\defn \{j\in \cF: j \le i-1\}$,
in (b) split $\cF$ into  
$$
\cM = \left\{i\in [1:n]: 2^{-n^\beta} <H(U_i|Y^nU^{i-1})\le 1-2^{-n^\beta}\right\}
$$
and
$$
\cH = \left\{i\in [1:n]: H(U_i|Y^nU^{i-1})> 1- 2^{-n^\beta}\right\},
$$
and in (c) used polarization results \cite{ArikanISIT2010} to write the bound $|\cM| = o(n)$.
Thus, by Lemma~\ref{Lemma:Constraints} and Lemma~\ref{Lemma:CM}, we conclude that the rate penalty $I(S;Y^n)$ is $o(n)$ and $I(X;Y)$ is achievable using the polar parity-check ensemble.

The number of constraints imposed by polar parity-checks is $|\cF|$, which is $nH(X|Y)+o(n)$ \cite{ArikanISIT2010}. The dimensionality of the ensemble $X^n$ is reduced from $nH(X)+o(n)$ to $nI(X;Y)+o(n)$ by the polar parity-checks; this is the smallest possible dimensionality (to order $O(n)$) for an ensemble that achieves $I(X;Y)$. 

We refrained from calling the codes generated under polar parity-checks ``polar codes'' because there are major differences between the two classes of codes.
To discuss this further, let us refer to the polar parity-check codes of this paper as PPC codes and reserve the term ``polar code'' for ordinary polar codes as defined in \cite{ArikanIT2009}.
The results of this paper establish that PPC codes achieve $I(X;Y)$ with a probability of error that goes to zero exponentially in $n$, while for polar codes
that exponent is not better than $\sqrt{n}$ even under ML decoding.
The $\sqrt{n}$ exponent arises from the fact that the minimum distance of a code generated by a submatrix of $G_\ell$ cannot have a minimum distance better than $O(\sqrt{n})$ for any fixed non-zero code rate.
It must be that on average PPC codes have a minimum distance proportional to $n$; otherwise, their error exponent would not be proportional to $n$.
This significant increase in minimum distance can be attributed to random selection of codewords; a PPC code may be seen as an expurgated polar code. 
The expurgation removes the defects in the polar code; but it also destroys the linear structure in the code. 
In standard polar coding, the mapping from messages to codewords is a linear relation of the form $x^n = u^nG_\ell$, which can be implemented in complexity $O(n\log(n))$.
Under PPC coding, there is no linear relationship of this type between data bits and codewords; hence, one can no longer claim that the encoding complexity is $O(n\log(n))$. 
Thus, PPC codes show a gain in performance at the expense of giving up the low-complexity encoding properties of polar codes.
Clearly, similar remarks apply to the complexity of decoding. 

For PPC codes, achieving $I(X;Y)$ under an arbitrary target distribution $q(x)$ is no different than achieving it under a uniform $q(x)$. 
With polar codes, achieving $I(X;Y)$ for a non-uniform $q(x)$ is not a straightforward task; it requires extension of the standard method and employing common randomness between the encoder and decoder in order to shape the channel input distribution \cite{Honda}. With PPC codes, the shaping is built into the code selection procedure.

\section{Summary}\label{Sect:Remarks}
The main motivation for this work has been to develop a packing lemma for polar codes that would enable translation of proofs by standard packing lemmas to similar results for polar coding.
More work needs to be done to accomplish this broader goal.
The main contribution of the paper has been the development of a technique for analyzing the performance of a random code ensemble defined by a fixed parity-check matrix.
In this sense, the results may have relevance to a broader class of codes than polar codes. 
An interesting observation in the paper has been that the polar parity-check ensemble shows markedly better performance than the standard polar code of the same size.
A better understanding of this phenomenon may be useful in designing better polar codes.
 
\section*{Acknowledgment}
This work was supported in part by the European Commission in the framework of the FP7 Network of Excellence in
Wireless COMmunications NEWCOM\# (contract n.318306).


\begin{thebibliography}{1}
 
\bibitem{EGK} A. El Gamal and Y.-H. Kim, {\sl Network {I}nformation {T}heory.} Cambridge University Press, 2011.

\bibitem{CM} D. Dubhashi and A. Panconesi, {\sl Concentration of Measure for the Analysis of Randomised Algorithms.} Cambridge University Press, 2009.

\bibitem{LDPC} R. G. Gallager, {\sl Low Density Parity Check Codes.} Monograph, M.I.T. Press, 1963.


\bibitem{ArikanIT2009} E.~Ar{\i}kan, ``Channel polarization: A method for constructing capacity-achieving codes for symmetric binary-input memoryless channels,''   
{\em IEEE Trans. Inform. Theory},vol.~55, pp.~3051--3073, July 2009.

\bibitem{ArikanISIT2010} E. Ar{\i}kan, ``Source polarization,'' in {\em Proc. 2010 IEEE Int. Symp. Inform. Theory}, (Austin, TX),  pp. 899-903, 13-18 June 2010.

\bibitem{Honda} J. Honda and H. Yamamoto, ``Polar coding without alphabet extension for asymmetric channels,'' in 2012 IEEE International Symposium on Information Theory Proceedings (ISIT), 2012, pp. 2147-2151.



\end{thebibliography}
\end{document}